\documentclass[12pt]{article}
\usepackage[utf8]{inputenc}
\usepackage[margin=1in]{geometry}
\usepackage{setspace}
\usepackage{amsmath,amssymb,graphicx,tikz,mathtools,pgfplots,hyperref,aliascnt,subfigure}
\DeclareMathOperator*{\argmax}{arg\,max}

\usepackage{amsthm}

\newtheorem{theorem}{Theorem}[section]
\newtheorem{lemma}{Lemma}[section]
\newtheorem{prop}{Proposition}

\newtheorem{defn}{Definition}

\usepackage{caption}
\usepackage{setspace}
\newcommand{\init}{\mathbf{x}}
\newcommand{\fin}{\mathbf{y}}
\newcommand{\final}{\mathbf{y}_{\pi}^*}
\newcommand{\K}{K}
\newcommand{\KK}{\tilde{K}}
\newcommand{\ws}{\ell_O}
\newcommand{\wrr}{\ell_R}
\newcommand{\I}{\mathcal{I}}
\newcommand{\ubar}[1]{\text{\b{$#1$}}}

\newcommand{\E}{\mathbb{E}}
\newcommand{\R}{\mathbb{R}}
\newcommand{\dist}{F}
\newcommand{\goal}{\mathcal{Y}}
\newcommand{\mean}{\mu}
\newcommand{\lol}[1]{\ubar{\theta}_{#1}}
\newcommand{\hil}[1]{\bar{\theta}_{#1}}

\newcommand{\curve}{\mathcal{Z}(\init)}
\newcommand{\norm}[1]{\|{#1}\|}
\pgfmathdeclarefunction{gauss}{2}{\pgfmathparse{1/(#2*sqrt(2*pi))*exp(-((x-#1)^2)/(2*#2^2))}%
}

\title{\LARGE
Optimal Information Provision for Strategic Hybrid Workers
}

\author{Sohil Shah, Saurabh Amin, Patrick Jaillet
\thanks{S. Amin is with the Laboratory for Information and Decision Systems, P. Jaillet is with the Department of Electrical Engineering and Computer Science, Laboratory for Information and Decision Systems, and Operations Research Center, S. Shah is with the Laboratory for Information and Decision Systems, and Operations Research Center, Massachusetts Institute of Technology (MIT), Cambridge, MA, USA, \{sshah95,amins, jaillet\}@mit.edu}}
\date{}
\begin{document}

\maketitle

\begin{abstract}
We study the problem of information provision by a strategic central planner who can publicly signal about an uncertain infectious risk parameter. Signalling leads to an updated public belief over the parameter, and agents then make equilibrium choices on whether to work remotely or in-person. The planner maintains a set of desirable outcomes for each realization of the uncertain parameter and seeks to maximize the probability that agents choose an acceptable outcome for the true parameter. We distinguish between stateless and stateful objectives. In the former, the set of desirable outcomes does not change as a function of the risk parameter, whereas in the latter it does. For stateless objectives, we reduce the problem to maximizing the probability of inducing mean beliefs that lie in intervals computable from the set of desirable outcomes. We derive the optimal signalling mechanism and show that it partitions the parameter domain into at most two intervals with the signals generated according to an interval-specific distribution. For the stateful case, we consider a practically relevant situation in which the planner can enforce in-person work capacity limits that progressively get more stringent as the risk parameter increases. We show that the optimal signalling mechanism for this case can be obtained by solving a linear program. We numerically verify the improvement in achieving desirable outcomes using our information design relative to no information and full information benchmarks. 

\end{abstract}

\section{Introduction}
\subsection{Motivation}
\noindent The COVID-19 pandemic sparked tremendous interest in practical tools to mitigate disease spread \cite{ji_lockdown_2020,nowzari_analysis_2016,drakopoulos_efficient_2014, chernozhukov_causal_2021}. One can distinguish between two types of non-pharmaceutical interventions central planners use: (i) hard and (ii) soft. 

\noindent \textit{Hard interventions} are rigorously enforced measures such as lockdowns, capacity limits or mask mandates. They prohibit certain actions that risk infectious spread. Balancing these interventions with their negative economic impact has recently been an active area of research \cite{acemoglu_optimal_2020,acemoglu_optimal_2021,birge_controlling_2020}. These measures are indeed effective at flattening the contagion curve, particularly in the early phases of a pandemic when cures and vaccines are unavailable \cite{anderson_how_2020}. However, these measures alone are not viable in the long-term as they tend to be economically and socially 
costly. 

\noindent On the other hand, \textit{soft interventions} aim to \textit{influence} agents to choose less risky actions. For example, a central planner can influence agents' choices by incentivizing safer alternative actions or penalizing riskier actions. Soft interventions are generally much less costly to implement than hard interventions, however their design is a nascent area of study \cite{de_vericourt_informing_2021,ely_rotation_2021,hernandez-chanto_contagion_2021}. Our work in this paper is motivated by the public health gains that can be achieved by strategic provision of information. The basic idea is that a central planner can generate stochastic, possibly coarse signals into the true value of an uncertain parameter that factors into agent utilities. The signals reshape agents' beliefs over this parameter and thereby influence their chosen actions to induce a socially desirable outcome. In practice, such signals can correspond to planner's broadcast on the virulence of a disease strain or reporting case counts to a particular granularity. 

\noindent Specifically, we focus on how to optimally provision information to strategic hybrid workers. A workforce has $\K$ groups of non-atomic agents, each deriving value from working in-person. However, there is a stochastic cost associated with becoming infected that increases with the number of in-person contacts and an unknown continuous stochastic parameter, $\theta \in \Theta$, measuring the disease's infectious risk. The central planner has imperfect information about $\theta$, but can publicly generate signals into the true value of $\theta$ using an information mechanism. On observing the signal, agents' public belief over $\theta$ is updated and they simultaneously make equilibrium choices on where to work. The planner is strategic in that it values a set of outcomes as desirable based on the true value of $\theta$ and some prescribed public health criteria. The novelty of our information provision problem (see Section \ref{sec:model}) is that these criteria only need to be described by a desirable set of outcomes for each parameter value. We distinguish between stateless and stateful objectives; in the former the set of desirable outcomes does not change as a function of the risk parameter, whereas in the latter it does. The central planner's objective is to maximize the probability that a desirable outcome is achieved. We investigate how to optimally provision information and evaluate its effectiveness to further the planner's objective.

\noindent In Section \ref{sec:static_id}, we derive the optimal signalling mechanisms for stateless objectives. We first characterize the equilibrium that agents will achieve as a function of their beliefs over $\theta$ in Proposition \ref{prop:eqbm_formula_prop_1}. In equilibrium, if there are agents working remotely from a group deriving higher benefit from in-person work than another group, then all agents in the other group also work remotely. This restricts the set of possible equilibria to a one-dimensional manifold in the $\K$-dimensional simplex. The intersection of this manifold and the desirable set of outcomes is precisely the set of outcomes the planner can hope to induce via information provision. In Lemmas \ref{lemma:m_smooth}, we characterize the mapping between posterior mean beliefs over $\theta$ and all achievable equilibria and in Theorem \ref{thm:stateless}; we derive the optimal mechanism. We find that the signalling depends on the position of the prior mean belief over $\theta$ relative to the subset of posterior mean beliefs in $\Theta$ that induce desirable outcomes. The mechanism uses at most two signals and partitions the parameter domain into at most two intervals with the signals generated according to an interval-specific distribution. This result adds to a growing literature on the optimality of monotone partitional signalling mechanisms \cite{dworczak_simple_2019,kolotilin_optimal_2018,guo_interval_2019}.

\noindent We next consider stateful outcome sets that enforce in-person capacity limits that progressively get more stringent as $\theta$ increases in Section \ref{sec:stateful}. For the sake of simplicity, we consider $\Theta$ to be discrete and identify a linear program that solves for the optimal objective achievable with a signalling mechanism (Proposition \ref{prop:stateful}). 

\noindent We then present in Section \ref{sec:numerical} a numerical comparison of our signalling mechanisms against two benchmarks: (i) revealing the true $\theta$ (full information) and (ii) revealing nothing (no information). Our signalling mechanisms significantly improve on these benchmarks for both stateless and stateful capacity objectives. Due to space
limitation all proofs are deferred to the appendix.

\subsection{Related Literature}
\noindent Information design has been active area of research in the economics community spawned by the seminal results shown by \cite{kamenica_bayesian_2011} (see \cite{candogan_information_2020, bergemann_information_2019, kamenica_bayesian_2019} for useful reviews). 

\noindent Our contribution adds to the recent literature on information provision for multiple agents over a continuous unknown parameter. The authors of \cite{gentzkow_rothschild-stiglitz_2016} show the equivalence between signalling mechanisms over continuous parameters and mean-preserving contractions of the parameter's prior distribution. This insight is useful for reducing the complexity of search over signalling mechanisms. The authors in \cite{candogan_optimal_2021} build on this idea to characterize an optimal recommendation system for a discrete number of agents that are uncertain about a system state and must reveal a privately held type. Their results characterize how the designer must choose his recommended actions to maximize his quasilinear objective over these actions while retaining incentive compatibility. We instead focus on a setting with non-atomic agents and a broader class of objectives that need only be represented by sets of outcomes. 

\noindent With regard to the design of soft interventions to mitigate disease spread, the authors of \cite{ely_rotation_2021} consider the optimal design of rotation schemes where in-person workers alternate on some schedule. Related to information provision, \cite{hernandez-chanto_contagion_2021} presents an activity-based model where the principal informs agents of the infection rate and identifies when full disclosure maximizes society's expected welfare. Most closely related to our work, \cite{de_vericourt_informing_2021} presents a macro-perspective on how central planners can provision information about unknown stochastic risk during a pandemic. Our setting is more general than \cite{de_vericourt_informing_2021}, as their central planner seeks to globally optimize over an objective function that takes a weighted average of functional measures of the economic health and infectious health. In contrast, our setup allows for a larger class of objectives that can incorporate other practically relevant factors. Our results also go beyond a binary distribution over the infectious risk by considering the unknown parameter as being drawn from any bounded, continuous distribution. There is also significant recent literature on infection dynamics in mean-field games \cite{olmez_modeling_2021,la_torre_mobility_2022}, however our paper focuses on information design after a stochastic representation of infectious risk has been obtained. Finally, the implications of our results for designing soft interventions are complementary to the works on containing infectious spread through hard interventions \cite{acemoglu_optimal_2020,acemoglu_optimal_2021,birge_controlling_2020,cianfanelli_lockdown_2021,wu_optimal_2020}.

\section{Model and Problem Formulation}
\label{sec:model}
\noindent We consider a model where strategic hybrid workers (agents) choose whether or not to work at a shared workspace amidst a pandemic of a communicable disease. Agents face uncertainty over the infectiousness of the disease (parameter), denoted by $\theta$, and thus over their likelihood of incurring the cost associated with becoming ill. In Sec. \ref{subsec:model_agent}, we describe agents' utilities and their dependence on the unknown parameter.

\noindent A central planner overseeing the agents maintains preferences over the aggregate outcomes of agents. The preferences may or may not depend on the true parameter. For a set of desirable agent outcomes for each value of the parameter, the planner seeks to maximize the probability that the equilibrium outcome lies in the set corresponding to the true parameter. Despite facing the same uncertainty over the parameter as the agents, the planner seeks to implement a mechanism that stochastically generates a publicly disclosed signal into the parameter's true value. The signal is revealed before agents make their choices, thereby influencing public belief over the parameter and affecting the equilibrium that is achieved. In Sec. \ref{subsec:model_signalling_cp}, we formalize this class of signalling mechanisms and identify the planner's problem of designing the optimal mechanism.
\subsection{Agent Behavior and Infectious Risk}
\label{subsec:model_agent}
\noindent We consider a unit mass of risk-neutral, Bayesian-rational agents. Each agent simultaneously commits to either working in-person at a shared workspace ($\ws$) or working remotely ($\wrr$). Each agent belongs to a group $k \in [\K] \coloneqq \{1,..,\K\}$. We represent the size of each group using the mass vector $\init \coloneqq \big(x_1, .., x_{\K}\big) \in \R_+^\K$ with the mass of agents belonging to each group $k$, $x_k$, being common knowledge. Agents' choices result in a \textit{remote mass} vector $\fin \coloneqq \big(y_1, .., y_{\K}\big)$, where $0 \leq y_k \leq x_k$ is the mass of agents from group $k$ who elected to work remotely. The \textit{in-person mass} vector is then $\init-\fin$, with $x_k - y_k \geq 0$ agents from group $k$ working in-person. Letting $\|\cdot\|$ correspond to the $L^1$-norm, note that $\|\init\| = 1 = \|\fin\| + \|\init-\fin\|$.

\noindent Any agent in any group $k$ who works remotely faces no infectious risk, but also receives no benefit from the work environment. Consequently, their utility $u_k(\wrr,\fin;\theta) = 0$ independent of $\theta$ and $\fin$. 

\noindent Agents from group $k$ that work in-person, however, each receive a commonly-known benefit $v_k$ where $v_k > 0$. The benefit encapsulates the personal value agents derive from working in-person net any travel costs. Without loss of generality, we assume that $v_k$ are distinct and we index groups so that $v_k$ is strictly decreasing in $k$. 

\noindent Agents that work in-person face both an uncertain risk of contracting the disease per infectious contact and also face uncertain costs in the event they do become ill. Both the transmissivity and severity of the disease are natural sources of public uncertainty as they cannot be easily estimated by individual agents. The overall risk agents can expect to incur directly depends on the in-person mass in a known way, as the increased frequency of contacts and likelihood of infected individuals being present can be easily estimated. Since agents are risk neutral, we capture the product of the uncertain risk per contact and uncertain costs into an uncertain \textit{infectiousness parameter} $\theta \in \Theta := [0,M]$ for some $M \in \R_+$. Critically, we assume that the agents cannot perceive the value of the parameter $\theta$ that was realized for this specific disease, but that the agents and planner all commonly believe a prior distribution $\dist$ over $\Theta$ from which $\theta \sim \dist$.

\noindent We model the (expected) stochastic infectious cost $\beta(\theta, \fin)$ to be linear in $\theta$ and decreasing in the remote mass (increasing for in-person mass). Specifically, for commonly-known, positive, differentiable functions $c_1, c_2: [0,1] \rightarrow \mathbb{R}_+$ with $c_1$ decreasing, $c_2$ non-increasing and $c_1(1) = c_2(1) = 0$:
\begin{align}
    \label{eqn:infectious_cost}
    \beta(\theta, \fin) \coloneqq \theta c_1(\|\fin\|)+c_2(\|\fin\|)
\end{align}
We further motivate this form of the infectious cost in greater detail in Appendix $\ref{appendix:aisys}$. Given the actions of all agents $\fin$ and the true parameter $\theta$, the utility of agents of group $k$ working in-person is:
\begin{align}
\label{eqn:utility_fxn}
    u_k(\ws,\fin;\theta) =  v_k - \beta(\theta, \fin)
\end{align} 

\subsection{Signalling and Central Planner Objective}
\label{subsec:model_signalling_cp}
\noindent The planner can implement a \textit{signalling mechanism} which publicly provisions information into the true realization of $\theta$ to all agents. Specifically, we assume that the planner publicly commits to and discloses a mechanism $\pi = \langle \{g_\theta\}_{\theta \in \Theta}, \I \rangle$ where $\{g_\theta\}_{\theta \in \Theta}$ are a set of probability distributions with each $g_\theta$ distributed over the set of \textit{signals} $\I$. If the true realization of the infectious risk is $\Tilde{\theta} \in \Theta$, then a signal $i \in \I$ is randomly generated with probability $g_{\Tilde{\theta}}(i)$. The planner does not observe the true parameter value before committing to $\pi$. Thus, the planner cannot influence the mechanism after the true parameter is realized.

\noindent Prior to choosing their actions, all agents view $i \in \I$ and symmetrically update their belief over $\theta$ to a posterior distribution $\dist_i$ where:
\begin{align}
    F_i(t) &= \mathbb{P}[\theta \leq t|i] = 
    \frac{\int_{0}^{t} g_\theta(i) d\dist(\theta)}{\int_0^M g_\theta(i) d\dist(\theta)} \label{eqn:F_i}
\end{align}

\noindent For any given $\pi$, it is useful to represent the corresponding probability with which each signal $i \in \I$ is generated and the posterior mean of $\theta$ upon viewing each signal $i \in \I$. If agents' strategies only depend on their mean beliefs over $\Theta$, this representation is equivalent to a \textit{direct mechanism} where the planner performs the update on the agents' behalf and simply shares the updated posterior mean with all agents. We denote such a mechanism by the set of tuples $\mathcal{T}_{\pi} = \{(q_i, \theta_i)\}_{i \in \mathcal{I}}$ where:
\begin{align}
    q_i &\coloneqq \int_{0}^{M} g_{\theta}(i) d\dist(\theta) \hspace{1em} [\text{signal probability}] \label{eqn:q_i}\\
    \theta_i &\coloneqq \frac{\int_{0}^{M} \theta g_{\theta}(i) d\dist(\theta) }{\int_{0}^{M} g_{\theta}(i) d\dist(\theta) }\hspace{1em}  [\text{posterior mean}] \label{eqn:theta_i}
\end{align}

\noindent A particularly relevant subclass of mechanisms is one in which  $\Theta$ is partitioned into intervals and the parameters in each interval collectively map to a single interval-specific signal. Namely, we call a mechanism $\pi$ \textit{monotone partitional} if there exists a finite partition of $\Theta$, $\mathcal{P} \coloneqq \{\Theta_i\}_{i=1}^m = \{[t_{i-1},t_i]\}_{i=1}^m$  for some $m$ with $t_0 = 0, t_m = M$ and some increasing sequence $\{t_i\}_{i=1}^{m-1} \subseteq \Theta$ such that $\I = [m]$ and for all $\theta$, $g_\theta(i) = \mathbb{I}\{\theta \in [t_{i-1}, t_i]\}$. Observe from \eqref{eqn:q_i} and \eqref{eqn:theta_i} that for these mechanisms $q_i = F(t_i) - F(t_{i-1})$ and $\theta_i = \frac{\int_{t_{i-1}}^{t_i} \theta dF(\theta)}{q_i}$. This structure is important in optimal information provision for continuous parameters as these interval-based signalling mechanisms correspond to the vertices of the polytope containing all feasible signalling mechanisms over the parameter \cite{dworczak_simple_2019,guo_interval_2019,candogan_optimal_2021,ivanov_optimal_2015}.

\noindent Each instance of our game for a particular initial mass vector $\init$ and public belief $G$ over $\Theta$ can be expressed in normal form $\Gamma(\init, G) = \langle \K, \init,(\mathcal{S}_k)_{k \in [\K]},(u_k)_{k \in [\K]}\rangle$:

$\K, \init$: $\K$ groups of agents, $x_k$ agents in each group

$\mathcal{S}_k$: group $k$'s agents' strategies, $\mathcal{S}_k = \{\ws,\wrr\}, \forall k \in [\K]$ 

$u_k$: group $k$'s agents' utility functions $u_k(s, \fin; \theta)$ for 

strategy $s \in \mathcal{S}_k$, aggregate strategy $\fin$ and parameter $\theta$

\noindent To implement the mechanism we consider the notion of a (Nash) equilibrium where no agents have an incentive to deviate from where they decided to work. In particular, agents of group $k$ at an equilibrium $\fin^*$ should always be choosing an action that generates utility equal to the best alternative of working remotely or in-person, $\argmax_{\ell \in \{\ws,\wrr\}}  \E[u_k(\ell, \fin^*; \theta)]$. In equilibrium, a group's agents spread across both actions if and only if this group's agents are indifferent between $\ws$ and $\wrr$. The following definition exploits this insight:

\begin{defn}
\label{defn:pse}
The remote mass vector $\fin^*$ is a Nash equilibrium of our game $\Gamma(\init,G)$ if and only if both:
\begin{itemize}
    \item[(i)] $\E_{\theta \sim G}[u_k(\ws, \fin^*; \theta)] > 0 \implies y_{k}^*(i) =  0$
    \item[(ii)] $\E_{\theta \sim G}[u_k(\ws, \fin^*; \theta)] < 0 \implies y_{k}^*(i) =  x_k$
\end{itemize} 
\end{defn}

\noindent Once signal $i \in \I$ is publicly revealed with the mechanism $\pi$, a signal-specific public belief $F_i$ over the risk parameter (computed using $\eqref{eqn:F_i}$) is induced and we denote the equilibrium remote mass that agents choose in response by $\final(i) \in \R^{\K}_+$.

\noindent In our model, the planner seeks to maximize the probability that the equilibrium that is achieved by the agents belongs to a set of desirable outcomes $\goal(\theta)  \subset \R^\K$ that depends on the true parameter $\theta$. The significance of this set for which we seek to maximize its attainability can arise out of practically relevant notions like capping the stochastic infection cost (i.e. the set $\goal(\theta) \coloneqq \{\fin: \beta(\theta,\fin) \leq \epsilon\}$) or maximizing total social surplus (i.e. the set $\goal(\theta) \coloneqq \{\fin: \sum_{i=1}^{\K}v_i(x_i-y_i) - (1-\norm{\fin})\beta(\theta,\fin) \geq \epsilon\}$). 

\noindent We focus explicitly on two classes of objectives. The first is what we denote as \textit{stateless} objectives where $\goal(\theta)$ does not depend on $\theta$. Hence, the objective set can be represented by a set $\goal$ independent of $\theta$. The second is a \textit{stateful} objective wherein we restrict the shapes of the objective sets to be nested as $\theta$ increases. Particularly, we focus on a \textit{scaled capacity} objective where $\goal(\theta) = \{\fin: \|\fin\| \geq b(\theta)\}$ for some increasing function $b(\cdot)$. From a practical viewpoint, the stateful objective mimics capacity restrictions that are imposed on the shared workspace which are gradually scaled back as the infectiousness decreases. 

\noindent We can express the (expected) objective $V(\pi)$ of the planner using $\pi$ by:
\begin{align}
    V(\pi) = \mathbb{P}\{\final(i) \in \goal(\theta)\} \label{eqn:objective_stateful}
\end{align}
For a stateless objective parameterized by $\goal$ this corresponds to $V(\pi) = \mathbb{P}\{\final(i) \in \goal\}$. 

\noindent The planner seeks to design $\pi^*$ so that: 
\begin{align}
    \pi^* &\in \argmax_{\pi: \langle \{g_\theta\}_{\theta \in \Theta},\I\rangle} V(\pi)
    \label{eqn:opt_v}
\end{align}

\section{Stateless Information Design}
\label{sec:static_id}
\noindent In this section, we focus on designing $\pi^*$ for a \textit{stateless} planner objective $\goal$. In Sec. \ref{subsec:static_eqn_char}, we compute the exact form of the equilibrium $\final(i)$ as a function of agent's posterior mean beliefs $\theta_i$. We then derive the structure of the optimal mechanism used by the planner in Sec. \ref{subsec:max_static}.

\subsection{Equilibrium Characterization}
\noindent Recall from \eqref{eqn:utility_fxn} that agents' utilities have a linear dependence on $\theta$. Since agents are risk-neutral, it is clear from Definition \ref{defn:pse} that the equilibrium outcomes $\final(i)$ should only vary as a function of the posterior means $\theta_i$. Observe that agents that work in-person incur the same expected costs but may derive different benefits, so the agents deriving the highest value from working in-person should do so in an equilibrium. The following proposition formalizes this insight and characterizes the mass vector $\final(i)$ and total in-person mass $m(\theta_i) \coloneqq 1-\norm{\final(i)}$ in equilibrium. 
\label{subsec:static_eqn_char}
\begin{prop}
\label{prop:eqbm_formula_prop_1}
Let $s_j = \sum_{i=1}^{j} x_i$ for all $j \in \{0,1,..,\K\}$, and $v(u) = v_j$ if $s_{j-1} \leq u < s_j$ for all $u \in [0,1]$. For any signal $i \in \I$ under mechanism $\pi$, the equilibrium in-person mass ${m(\theta_i)  = \sup\{u: v(u) \geq c_1(1-u)\theta_i+c_2(1-u)\}}$ and:
\begin{align*}
    (\final(i))_k = \begin{cases}
0 & \text{ if } s_k \leq m(\theta_i) \\
m(\theta_i) - s_{k-1} & \text{ if } s_{k-1} < m(\theta_i) < s_k \\
x_k & \text{ if } s_{k-1} \geq m(\theta_i)
\end{cases}
\end{align*}
\end{prop}

\noindent Intuitively, the proposition sets the in-person mass $m(\theta_i)$ by having as many agents work in-person as possible until the remaining agents find the infectious cost too high. The non-increasing function $v(u)$ corresponds to the highest benefit amongst the mass of $1-u$ agents with the smallest benefits $v_i$. The computation merely increases $u$ until $v(u)$ no longer exceeds the expected infectious costs.

\begin{figure}[h!]
\centering
\includegraphics[width=45mm]{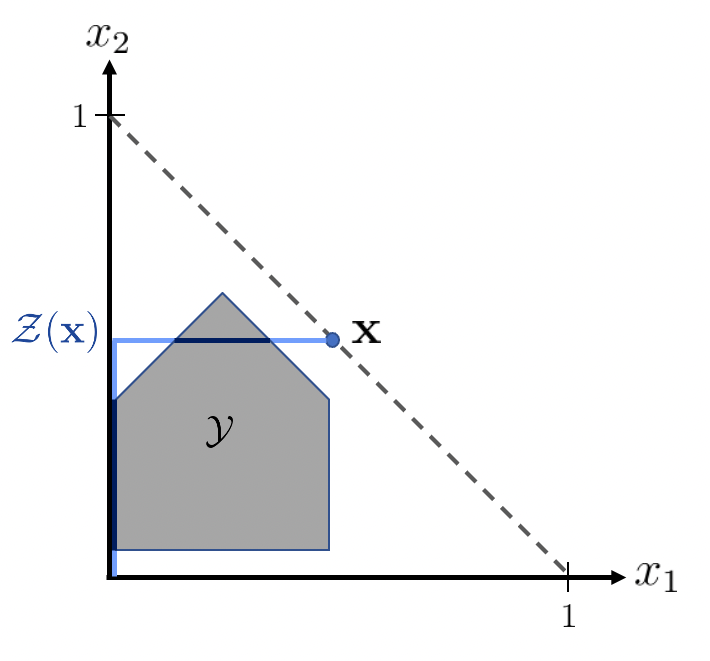}
\caption{Example of $\mathcal{Z}(\init)$ (solid curve in blue) when $\K = 2$.}
\label{fig:rec_curve_ex}
\end{figure}
\noindent Proposition $\ref{prop:eqbm_formula_prop_1}$ also identifies a threshold-based structure that the equilibrium must obey. Specifically, we can define a critical group $k^*(i) = \inf\{k: s_k > m(\theta_i)\}$ where agents in groups $k$ for $k < k^*(i)$ all work in-person and agents in groups $k$ for $k > k^*(i)$ all work remotely. This characteristic determines a strict subset of all possible remote mass vectors $\fin$ that are achievable in equilibrium. Hence, from Proposition \ref{prop:eqbm_formula_prop_1}, we can conclude that for any $\pi$, any equilibrium $\final(i)$ must lie on a one-dimensional manifold ${\curve \coloneqq \{c\mathbf{e}_k + \sum_{i=k+1}^\K x_i \mathbf{e}_i| k \in [\K], c \in [0,x_k]\}}$ where $\mathbf{e}_i \in \R^\K$ with $(\mathbf{e}_i)_k = 1$ if $i=k$ and $0$ otherwise; see Fig. \ref{fig:rec_curve_ex}.

\noindent Moreover, as the posterior mean belief of the infectiousness $\theta_i$ increases, agents are more cautious to choose in-person work because the infectious cost increases on expectation. The following lemma establishes that the equilibrium in-person mass is monotone and smooth in $\theta_i$. 

\begin{lemma}
$m(\theta)$ is non-increasing, bounded and continuous in $\theta$.
\label{lemma:m_smooth}
\end{lemma}

\noindent Using Proposition \ref{prop:eqbm_formula_prop_1} and Lemma \ref{lemma:m_smooth}, we conclude that as the posterior mean $\theta_i$ increases, the equilibrium mass vector moves further from the initial point $\init$ and closer towards the origin along $\curve$. As a result, the design of the optimal signalling mechanism can be viewed as the problem of \emph{maximizing the probability} that \emph{posterior means} generated by the mechanism correspond to points $\fin \in \curve$ that the \emph{planner deems desirable} (i.e. $\fin \in \goal$).

\subsection{Optimal Information Design}
\label{subsec:max_static}
\noindent We now use Proposition \ref{prop:eqbm_formula_prop_1} to solve the planner's problem of choosing an optimal signalling mechanism $\pi$. The manifold $\curve$ contains all possible equilibria $\final(i)$ that could possibly be achieved for any $\pi$ and any realized signal $i$. However, the manifold need not intersect with the goal set $\goal$ at all, and in this case, no $\pi$ can induce an equilibrium $\final(i) \in \goal$. Thus, we need to exploit the geometric relationship between $\goal$ and $\curve$ to solve this design problem. 

\noindent More generally, in relation to well-behaved sets $\goal$, our next result shows that the equilibria in $\curve \cap \goal$ corresponds to a finite number of intervals of in-person mass. Specifically, we make the following assumption: \\
\noindent \textbf{(A1)} $\goal$ is closed and convex. \\
Moreover, we denote $k(\init,u) = \min\{j > 0: s_j > u\}$ as the lowest group that has not completely chosen in-person work after the mass of $u$ agents with highest benefit choose in-person work. Likewise, $z(\init,u)$ corresponds to the point on $\curve$ after $u$ agents choose in-person work or equivalently $z(\init, u) = (u-\sum_{k = 1}^{k(\init,u)-1} x_k)\mathbf{e}_{k(\init,u)} +  \sum_{k = k(\init,u)+1}^{\K} x_k \mathbf{e}_k$. Proposition \ref{prop:eqbm_formula_prop_1} implies that $z(\init, 1-\norm{\final(i)}) = \final(i)$ for all $i$. 
\begin{lemma}
Under Assumption A1, there exists $\tilde{\K} \leq \K$ closed and disjoint intervals $\Omega_i = [\omega_i^1, \omega_i^2] \subset [0,1]$ for all $i \in [\tilde{\K}]$ with $\goal \cap \curve = \cup_{i=1}^{\tilde{\K}} \{z(\init,u): u \in \Omega_i\}$.
\label{lemma:intervals}
\end{lemma}

\noindent Observe that if $\goal \cap \{z(\init,u): s_{k-1} \leq u \leq s_k\}$ is empty, then $\omega_k^1 > \omega_k^2$. Trivially, notice that if $\goal$ is contained in the interior of the positive orthant of $\R^\K$, then there is at most only one non-empty segment of intersection $\Omega_1$ since $\{z(\init,u): s_{k-1} \leq u \leq s_k\} \cap \text{int}(\R_+^\K) = \emptyset$ for all $k > 1$. We note that Lemma \ref{lemma:intervals} is generalizable beyond Assumption A1. Namely, if $\goal$ is nonconvex, our results are valid as long as the $\goal \cap \curve$ is able to be represented by a finite number of intersecting line segments $\Omega_i$.

\noindent Using Lemma \ref{lemma:intervals}, we can now re-express \eqref{eqn:opt_v} to instead depend on whether the remote masses lie in one of the intervals $\Omega_i$:
\begin{align}
        V^* &= \max_{\pi: \langle \{g_\theta\}_{\theta \in \Theta},\I\rangle} \mathbb{P}\{\final(i) \in \goal\} \nonumber\\
         &= \max_{\pi: \langle \{g_\theta\}_{\theta \in \Theta},\I\rangle} \mathbb{P}\{\final(i) \in \goal \cap \curve\} \nonumber\\
         &= \max_{\pi: \langle \{g_\theta\}_{\theta \in \Theta},\I\rangle} \sum_{k=1}^{\KK} \mathbb{P}\{1-\norm{\final(i)} \in \Omega_i\} \label{eqn:value_c_i} \nonumber\\       
         &= \max_{\pi: \langle \{g_\theta\}_{\theta \in \Theta},\I\rangle} \sum_{k=1}^{\KK} \mathbb{P}\{\omega_k^1 \leq m(\theta_i) \leq \omega_k^2\}           
\end{align}
From Lemma $\ref{lemma:m_smooth}$, we note that since $m$ is monotone, bounded and continuous, the preimage of this mapping over the closed interval $[\omega_k^1,\omega_k^2]$ is another closed interval $[\lol{k}, \hil{k}]$ where $m(\theta_i) \in \Omega_i$ if and only if $\theta_i \in [\lol{k}, \hil{k}]$. Without loss of generality, since the intervals $\Omega_i$ are disjoint, we let $\lol{k}$ be increasing in $k$ which also implies that $\hil{k}$ is increasing in $k$. Exploiting this fact and using a mechanism $\pi$'s corresponding direct mechanism $\mathcal{T}_\pi$ as computed in \eqref{eqn:q_i} and \eqref{eqn:theta_i}, we can express the planner's program as:
\begin{align}
        V^* &= \max_{\mathcal{T}_\pi: \{(q_i, \theta_i)\}_{i \in \I}} \sum_{k=1}^{\KK} q_i \mathbb{I}\{\lol{k} \leq \theta_i \leq \hil{k}\}           
        \label{eqn:red_1_form}
\end{align}

\noindent We make some immediate observations about this result. Firstly, observe that if the objective set $\goal$ does not intersect with the manifold $\curve$, then all the $\Omega_i$ are empty and hence the objective in \eqref{eqn:value_c_i} is necessarily equal to 0. By contrast, if $\E_{\theta \sim \dist}[\theta] \coloneqq \mean \in [\lol{l},\hil{l}]$ for some $l$, observe that a non-informative mechanism $\pi$ has the property that $\I$ is a singleton and $\mathcal{T}_\pi = \{(1,\mean)\}$. Hence, the planner achieves the maximal objective of $1$ with the non-informative mechanism. However, if $\mean \notin [\lol{l},\hil{l}]$ for any $l$, then the non-informative mechanism achieves an objective of 0. Consequently, we find it useful to categorize distributions $\dist$ based on the positioning of the prior mean $\mean$ relative to the intervals $\{[\lol{k},\hil{k}]\}_{k \in [\KK]}$ and solve for the optimal mechanism in each regime. In order to characterize these regimes, consider the following increasing functions $\ubar{s}(t) = \E[\theta|\theta \leq t]$ and $\bar{s}(t) = \E[\theta|\theta \geq t]$ which are computable from $\dist$. Then, we can define $\delta(\alpha,\lambda, t) \coloneqq \frac{ \ubar{s}(t) \lambda F(t) +  \bar{s}(t) \alpha (1-F(t))}{\lambda F(t) +  \alpha (1-F(t))}$ to represent the posterior mean belief when receiving a signal that is generated with probability $\lambda$ for $\theta \leq t$ and with probability $\alpha$ for $\theta > t$. We now introduce the following regimes and subregimes on the prior distribution $\dist$:
\begin{align*}
&\text{(R1): } \exists k \text{ such that } \mean \in [\lol{k},\hil{k}] \\
&\text{(R2): } \exists k \text{ such that } \mean \in (\hil{k},\lol{k+1}) \\
&\hspace{5mm}\text{(R2a): } \exists \ell, k \in [\KK], t \in \Theta, \alpha, \lambda \in [0,1] \\ &\quad \quad \quad \quad \text{ s.t. } \delta(\alpha,\lambda, t) \in [\lol{k},\hil{k}] \\ &\quad \quad \quad \quad \text{ and } \delta(1-\alpha,1-\lambda, t) \in [\lol{\ell},\hil{\ell}]\\ 
&\text{(R3): } \mean > \hil{\KK} \\
&\text{(R4): } \mean < \lol{1}
\end{align*}

\noindent The above regimes are illustrated in Fig. \ref{fig:result}. (R1) captures the aforementioned case where the prior distribution's mean $\mu$ already lies in an interval. Both (R3) and (R4) capture distributions where $\mu$ lies outside the convex hull of the intervals. (R2) considers when $\mu$ does not lie in any interval, but falls in the gap between two intervals. Intuitively, in this case the planner seeks to spread mean beliefs outwards from the gap. However, this information design may not be feasible for distributions $\dist$ that are tightly concentrated in this gap and $\delta(\cdot,\cdot,\cdot)$ will not take on many values outside this gap. Consequently, we focus on a more likely case -- subregime (R2a) -- where $\dist$ has sufficient probability mass outside the gap between two intervals surrounding $\mean$.

\noindent Before solving for an optimal mechanism in each regime, note that any direct mechanism $\mathcal{T}_\pi$ that is a solution to \eqref{eqn:red_1_form} is not unique in general. This is immediate as any signal $i \in \I$ can be forked into two signals $i_1, i_2$ uniformly at random and induce symmetric posterior means $\theta_i = \theta_{i_1} = \theta_{i_2}$ with probability $\frac{q_i}{2}$ each and hence achieve the same objective $V^*$. To eliminate such multiplicity, we provide the following lemma: it shows that we need not consider mechanisms that use a set of signals $\I$ of size $|\I| > \KK +1$ since we require at most one signal to correspond to a posterior mean in each of the $\KK$ intervals and one extra signal whose posterior mean lies outside these sets. 
\begin{lemma}
\label{lemma:num_signals_bdd}
There exists an optimal mechanism $\pi^* = \langle \{g^*_\theta\}_{\theta \in \Theta}, [\KK+1] \rangle$ where $V(\pi^*) = V^*$ such that for each $k \in [\KK]$, $\theta_k \in [\lol{k}, \hil{k}]$ and $\theta_{\KK+1} \notin [\lol{k}, \hil{k}]$ for all $k \in [\KK]$.
\end{lemma}
The insight given by Lemma \ref{lemma:num_signals_bdd} consequently reduces the search over all possible mechanisms $\pi$ to those that use at most $\KK+1$ signals and can be represented by a $\mathcal{T}_\pi$ with at most $\KK+1$ tuples. Using this reduction, we can without loss of optimality restrict the problem \eqref{eqn:red_1_form} to simpler signalling mechanisms $\mathcal{T}_\pi$ with less than $\KK+1$ tuples. Notice that if we had simply specified $\{(q_i, \theta_i)\}_{i \in [\KK+1]}$ instead of deriving it from a particular $\pi$, we might be optimizing over tuples that cannot actually be implemented by a mechanism $\pi$. To ensure implementability, we require that the distribution over the posterior means $G$ specified by the tuple with cumulative distribution $G(t) = \sum_{i=1}^{\KK+1} q_i \mathbb{I}\{\theta_i \leq t\}$ is a mean-preserving contraction of $\dist$. The equivalence between the distributions $G$ that are \textit{mean-preserving contraction} of $\dist$ and the distributions over posterior means that are implementable through a mechanism has been well-established \cite{blackwell_theory_1954, gentzkow_rothschild-stiglitz_2016}. Formally, $G \preccurlyeq \dist$ is a mean-preserving contraction if and only if $\forall t \in \Theta$, $\int_{0}^{t} F(t)dt \leq \int_{0}^{t} G(t)dt$.

\noindent We now exploit the above reduction and solve for the optimal mechanism that solves $\eqref{eqn:red_1_form}$. Define the increasing function $h(\theta) \coloneqq \sup\{s:\int_0^s F^{-1}(t)dt \leq s \theta\}$. The following theorem characterizes the optimal objective and optimal mechanism across R1, R2a, R3 and R4 (we discuss the optimal direct signalling mechanism for R2 in the appendix).

\begin{theorem}
\label{thm:stateless}
The optimal objective $V^*$ and a corresponding optimal mechanism can be computed as follows for each regime:\\
(R1): 
$V^* = 1$ and $\pi^*$ is monotone partitional with $t_0 = 0$ and $t_1 = M$  \\
(R2a):
$V^* = 1$ and $\pi^*$ has $\I = \{1,2\}$,  $g_\theta(1) = \lambda \mathbb{I}\{\theta \in [0, t]\}+\alpha \mathbb{I}\{\theta \in (t, M]\}$ and $g_\theta(1) = (1-\lambda) \mathbb{I}\{\theta \in [0, t]\}+(1-\alpha) \mathbb{I}\{\theta \in (t, M]\}$.\\
(R3):
$V^* = q_1^* \coloneqq \min\{h(\hil{\KK}),\frac{M-\mu}{M-\hil{\KK}}\}$ and $\pi^*$ is monotone partitional with $t_0 = 0$, $t_1 = F^{-1}(q_1^*)$ and $t_2 = M$. \\
(R4):
$V^* = 1-q_2^*$ where $q_2^* \coloneqq \inf\{q \geq \frac{\lol{1}-\mu}{\lol{1}}:q \leq h(\lol{1}-\frac{\lol{1}-\mean}{q}) \}$ and $\pi^*$ is monotone partitional with $t_0 = 0$, $t_1 = F^{-1}(q_2^*)$ and $t_2 = M$.
\end{theorem}
 Notice that in subcase (R1), an optimal mechanism is trivially monotone partitional with a single signal; in Figure \ref{fig:result}, the entire probability mass over $\Theta$ maps to a single signal. In (R3) and (R4) an optimal mechanism is also monotone partitional with two signals, where a threshold is specified by the theorem and the $\Theta$ is partitioned into a ``low" and ``high" interval, where each interval corresponds to a signal. In (R2b), we create two signals and identify a threshold $t$ where we split $\Theta$ into two intervals at that threshold. In each interval, we then reveal the first signal with probability $\alpha$ if $\theta \leq t$ and with probability $\lambda$ if $\theta > t$. Otherwise, we reveal the second signal. Consequently, the signals are generated according to an interval-specific distribution over just two signals. The analysis of these cases show that simple signalling mechanisms that use at most two signals are sufficient to achieve the optimum. 

\begin{figure}
    \centering
\includegraphics[width=80mm]{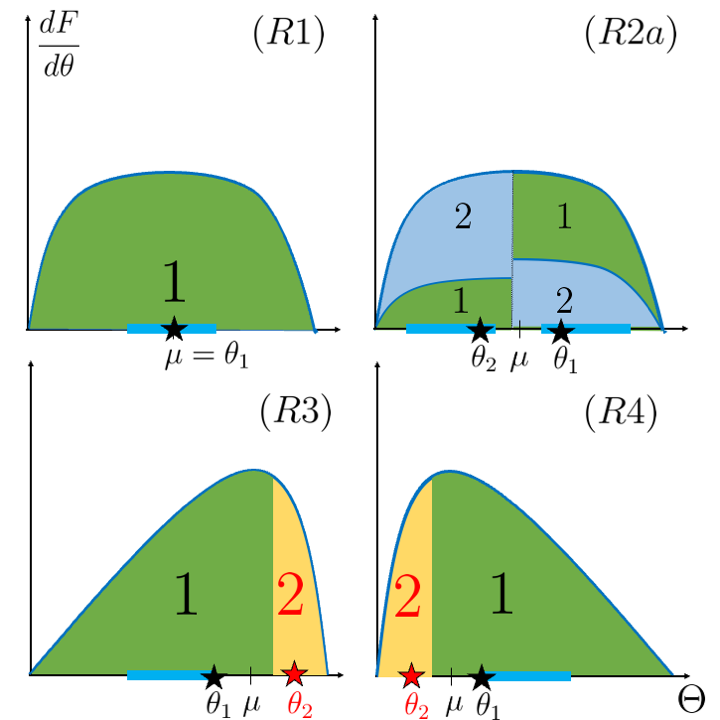}
    \caption{The pdf for $\theta$ and the intervals $[\lol{k},\hil{k}]$ (blue) with the solution for the four cases from Theorem $\ref{thm:stateless}$ marked by the induced posterior means $\theta_1, \theta_2$ (star). The posteriors are induced by mapping the probability mass to signals as marked (green, violet, yellow).}
    \label{fig:result}
\end{figure}

\section{Stateful Information Design}
\label{sec:stateful}
\begin{figure}[h]
\centering     
\subfigure[Goal sets for stateful objective when $\K = N = 2$ ($\nu_1 < \nu_2$).]{\label{fig:a}\includegraphics[width=60mm]{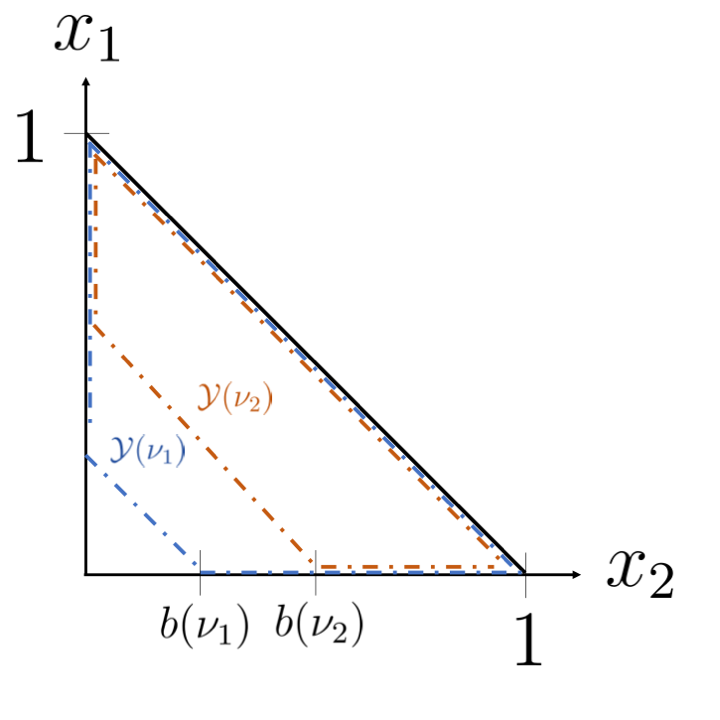}}
\subfigure[A more general nested structure of goal sets. ]{\label{fig:b}\includegraphics[width=60mm]{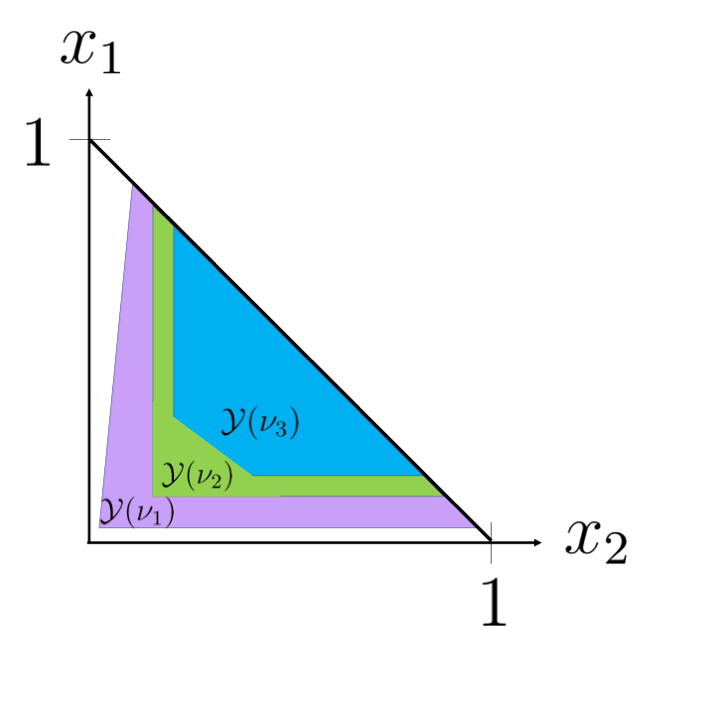}}
\caption{}
\end{figure}
\noindent In this section, we consider a stateful scaled capacity objective $\{\goal(\theta)\}_{\theta \in \Theta}$ that is defined by a collection of sets $\goal(\theta) = \{\fin: \|\fin\| \geq b(\theta)\}$ for some increasing function $b(\cdot)$ in the domain $\Theta$. Such sets are practically relevant because they model a capacity limit on in-person work that gets more stringent as the infectious risk increases. In this stateful setting, we cannot reduce the search over $\pi$ by searching over the corresponding direct mechanisms $\mathcal{T}_\pi$. This is because a posterior mean $\theta_i$ will induce a specific equilibrium outcome $\final(i)$ but we cannot evaluate whether $\final(i) \in \goal(\theta)$ if we only know $\theta_i$ and not the true $\theta$.

\noindent To proceed with analyzing the stateful case, for some $N > 0$ we restrict $\Theta \coloneqq \{\nu_1, \nu_2, .., \nu_N\}$ with $0 < \nu_1 < .. < \nu_N$. We denote $\mathbb{P}[\theta = \nu_i] = p_i$ for all $i \in [N]$. To simplify notation, we also denote $b(\nu_i) = b_i$ for all $i \in [N]$ where $b_i$ is strictly increasing in $i$; see Fig. \ref{fig:a}. 

\noindent Using \eqref{eqn:objective_stateful}, we can express the objective of the planner by: 
\begin{align}
    V^* &= \max_{\langle \{g_\nu\}_{\nu \in \Theta},\I\rangle} \mathbb{P}\{\final(i) \in \goal(\nu)\} \\
     &= \max_{\langle \{g_\nu\}_{\nu \in \Theta},\I\rangle} \sum_{j=1}^N \sum_{i=1}^{|\I|} p_j g_\nu(i) \mathbb{I}\{\final(i) \in \goal(\nu_j)\}  \label{eqn:stateful_a}
\end{align}
As with the stateless design, $\final(i)$ only depends on the induced posterior mean $\theta_i$ and $\final(i) \in \goal(\nu_j)$ if and only $\final(i) \in \goal(\nu_j) \cap \curve$. The latter follows if and only if $b_j \leq \|\final(i)\| \leq 1$ or $0 \leq m(\theta_i) \leq 1-b_j$. Using Proposition \ref{prop:eqbm_formula_prop_1} and Lemma \ref{lemma:m_smooth}, observe that this corresponds to $\theta_i \in [\gamma_j,\infty)$ for some constants $\gamma_j$ that are strictly increasing in $j$. For completeness, we denote $\gamma_0 = 0$ and $\gamma_{N+1} = \infty$.
We can then reformulate \eqref{eqn:stateful_a} as follows:
\begin{align}
    V^* &= \max_{\langle \{g_\nu\}_{\nu \in \Theta},\I\rangle} \sum_{j=1}^N \sum_{i=1}^{|\I|} p_j g_{\nu_j}(i) \mathbb{I}\{\gamma_j \leq \theta_i < \infty\} \label{eqn:stateful_b}
\end{align}
We highlight that despite these simplifications, our analysis is extensible to other practically relevant objectives aside from the scaled capacity objective. Suppose that the desirable sets satisfy Assumption A1 and are also (i) nested with $\goal(\nu_{i+1})\subset \goal(\nu_{i})$ for all $i \in [N-1]$, and (ii) half-bounded with $\init \in \goal(\nu_i)$ for all initial mass vectors $\init$ and all $i \in [N]$; see Fig. \ref{fig:b}. Then we arrive at an identical formulation to \eqref{eqn:stateful_b} since the minimal posterior mean beliefs for the agent outcomes to enter each $\goal(\nu_j)$ is an increasing sequence $\gamma_j$.

\noindent Using a similar reduction as in stateless design, we can restrict our mechanism to induce a single posterior mean $\theta_i$ to lie in each interval $[\gamma_{i-1},\gamma_{i})$ for $i \in [N+1]$. These $\theta_i$ consequently induce equilibrium outcomes that lie in $\goal(\nu_j)$ if and only if $\theta_i \geq \gamma_j$ or equivalently $i \geq j+1$. This implies that an optimal mechanism can be described by at most $N+1$ signals. In the proposition below, we provide a linear program (LP) to compute an optimal signalling mechanism for the stateful scaled capacity objective. We first introduce the LP:
\begin{align*}
    V^* &= \max_{\{ z_{ji}\}} \sum_{j=1}^N \sum_{i=j+1}^{N+1} z_{ji} \\
    &\text{s.t. }\sum_{i=1}^{N+1}  z_{ji} = p_j, \hspace{0.5cm}\forall j\in[N]\\
    &\hspace{1cm} 0 \leq z_{ji} \leq p_j,\hspace{0.5cm} \forall i\in[N+1],j\in[N] \\
    &\hspace{1cm} \gamma_{i-1}\sum_{j=1}^N z_{ji} \leq  \sum_{j=1}^N \nu_j z_{ji}, \hspace{0.5cm} \forall i\in[N+1] \\
    &\hspace{1cm} \sum_{j=1}^N \nu_j z_{ji} \leq \gamma_{i}\sum_{j=1}^N z_{ji}, \hspace{0.5cm} \forall i\in[N+1]
\end{align*}

\begin{prop}
\label{prop:stateful}
The optimal value $V^*$ to the above LP is equal to $\max_{\pi: \langle \{g_\theta\}_{\theta \in \Theta},\I\rangle} V(\pi)$ in \eqref{eqn:objective_stateful} under the stateful scaling capacity objective. For an optimal solution $\{z_{ji}^*\}_{j \in [N], i\in[N+1]}$ of the LP, the optimal signalling mechanism is given by $\pi^* = \langle \{g_\nu\}_{\nu \in \Theta},\I\rangle$ where $\I \coloneqq [N+1]$ and $g_{\nu_j}(i) = \frac{z_{ji}^*}{p_j}$ for all $i \in \I,j \in [N]$. 
\end{prop}
Using this result, we can efficiently compute the optimal objective $V^*$ and an optimal signalling mechanism $\pi^*$ in polynomial time. 

\noindent If the $b_j$ and hence $\goal(\nu_j)$ were all identical, our setting becomes the stateless setting solved by Theorem $\ref{thm:stateless}$. In this special case, all $\nu_j$ correspond to a single $\gamma=\gamma_j$ and only two posterior means $\theta_1 < \gamma \leq \theta_2$ are required. If $\mu = \sum_{j \in [N]} p_j \nu_j \geq \gamma$ then this corresponds to regime R1 and non-informative mechanism is optimal. If $\mu < \gamma$ this corresponds to regime R4 and likewise we recover a discretized interval-based signalling mechanism over $\I = \{1,2\}$ by solving the LP and we can check that there exists a $\hat{j} \in [N]$ such that for all $j < \hat{j}$, $g_{\nu_j}(1) = 1$ and for all $j > \hat{j}$, $g_{\nu_j}(2) = 1$.  

\noindent Finally, for any mechanism $\pi$, the probability of complying with the capacity limit conditioned on $\theta = \nu_j$ can be represented by $V_j(\pi) \coloneqq  \mathbb{P}\{\final(i) \in \goal(\theta)|\theta=\nu_j\}$. Observe that $V(\pi) = p_j V_j(\pi)$ is a weighted average of the compliance for each possible realization of $\theta$. However, practically planners may face significantly more adverse risks of being overcapacitated as the infectiousness increases, and weight compliance heterogeneously across realizations of $\theta$. For this slightly more general objective, an optimal signalling mechanism for a different weighted sum $\sum_{j=1}^{N} \alpha_j V_j(\pi)$ for $\alpha_j \geq 0$ can likewise be solved using the linear program above by replacing the objective with $\max_{\{ z_{ji}\}} \sum_{i=j+1}^{N+1} \frac{\alpha_j}{p_j}z_{ji}$.
\section{Numerical Experiments}
\label{sec:numerical}
\noindent We now present numerical experiments to illustrate the benefit of our optimally designed signalling mechanisms relative to no information and full information benchmarks. We consider a stateless case where the in-person safe capacity the planner seeks to implement does not depend on $\theta$. We vary the capacity limit to demonstrate how the relative improvement with regard to each benchmark changes. 

\noindent We implement an experiment where $\goal = \{\fin : \norm{\fin} \geq b\}$ and we vary the constants $b$ between 0 and 1 in increments of $\frac{1}{20}$. The number of groups is fixed at $\K = 10$ and for each $b$ we run $10,000$ simulations by drawing $ \mathbf{x} \sim Unif(L_1), v_i \sim Unif[0,10], \theta \sim Unif[5,20]$ ($\mu = \frac{25}{2}$). We evaluate the optimal mechanism given by Theorem \ref{thm:stateless} against two benchmarks. The first benchmark is the non-informative mechanism where agents receive no added information and consequently have a posterior mean belief of $\mu$ over $\Theta$. The second benchmark is the fully-informative mechanism, where the true value of $\theta$ is directly revealed so agents maintain posterior mean beliefs $\theta'$ where $\theta'$ is sampled from $Unif[5,20]$. 
The resulting objective that measures the probability of compliance is averaged across simulations for each $b$ and plotted in Figure \ref{fig:stateless_carying_capacity}. 
\begin{figure}
    \centering
    \includegraphics[width=75mm]{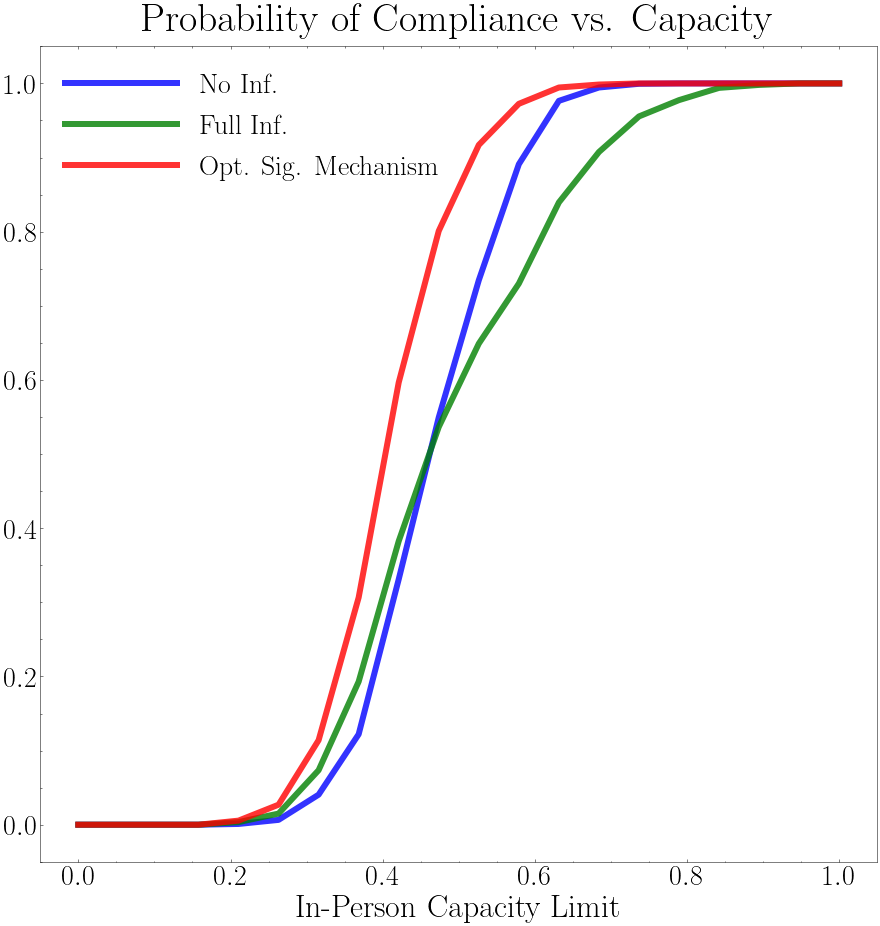}
    \caption{Comparison between capacity compliance across non-informative, fully-informative and optimal signalling.}
    \label{fig:stateless_carying_capacity}
\end{figure}
We indeed obtain that the optimal signalling mechanism strictly improves relative to both benchmarks. The improvement reduces as $b \nearrow 1$ as the intersection between the outcomes achievable in equilibrium and the outcomes the planner deems acceptable reduces to the empty set. Naturally, the mechanism's ability to influence agents to a set of outcomes decreases as the set of outcomes becomes smaller. Likewise, as $b \searrow 0$, the set of acceptable outcomes grows to encompass all outcomes. Hence, the mechanism approaches full compliance, but so do the benchmark mechanisms. Consequently, the relative value of our optimally-designed signalling mechanism decreases as the set of outcomes the planner is willing to accept grows. 
\noindent For the stateful capacity case, we highlight a specific instance where we scale-down the in-person capacity limit as the infectious risk increases. Consequently, the minimal posterior mean to induce outcomes in $\goal(\nu_j)$ are represented by $\gamma_j$ that are increasing in $j$. To gain intuition, we choose the specific set of tuples for $N=3$ where $\{p_i\}_{i \in \{1,2,3\}} = (0.3,0.3,0.4), \{\nu_i\}_{i \in \{1,2,3\}} = (0.4,0.6,1.0), \{\gamma_i\}_{i \in \{1,2,3\}} = (0.5, 0.9, 1.2)$ and compare the performance of the optimal mechanism against the same two benchmarks as in the stateless case.
In Fig. \ref{table:stateful_vary_capacity}, we tabulate each information mechanism's performance. Observe that under no information, agents mean belief over $\theta = \mu = \sum_{j=1}^N p_j \nu_j = 0.7$. Since agents receive no signal, the equilibrium outcome is independent of $\theta$ and agents only comply if $\gamma_j < \mu$ which happens with probability $p_1 = 0.3$. Under full information, agents perfectly observe $\theta$, so the outcome is compliant with the limit if and only if $\theta_j \geq \gamma_j$. This occurs with probability $0$ in our example. Thus, the stateful design significantly improves the probability of compliance relative to both benchmarks. We also observe that the stateful design does no worse in this setting than either benchmark in any of the conditional compliance measures $V_j(\pi)$.
\begin{figure}
    \centering
    \begin{tabular}{ |c|c|c|c|c| } 
    \hline
    Information Scheme & $V(\pi)$ & $V_1(\pi)$ & $V_2(\pi)$ & $V_3(\pi)$ \\     
    \hline
    No Information & 0.300 &1& 0& 0 \\ 
    Full Information  & 0.000 & 0& 0 &0 \\ 
    Stateful Design & 0.425 & 1& 0.417 &0 \\ 
 \hline
\end{tabular}
    
    \caption{Comparison of non-informative, fully informative and optimal stateful signalling in terms of the overall probability of compliance and conditional probabilities of compliance. }
    \label{table:stateful_vary_capacity}
\end{figure}

\section{Concluding Remarks}
\noindent In this paper, we introduced a model to study information provision for strategic hybrid workers. The central planner seeks to control the mass of in-person workers across each group in the equilibrium outcome. Our model
captures two key features: (a) a general objective that aims to maximize the probability that the equilibrium outcome lies in a particular set which may or may not be state-dependent; (b) heterogeneous workers making strategic decisions to trade-off in-person work and infectious risk. We provided a complete description of the equilibria of the game in response to the signals and derived the optimal signalling mechanism that the planner can employ. These results suggest guidelines for the design and deployment of signalling mechanisms as hard intervention measures are being phased down by public health agencies.

\bibliographystyle{IEEEtran}
\bibliography{ref} 

\begin{thebibliography}{10}
\def\url#1{}
\csname url@samestyle\endcsname
\providecommand{\newblock}{\relax}
\providecommand{\bibinfo}[2]{#2}
\providecommand{\BIBentrySTDinterwordspacing}{\spaceskip=0pt\relax}
\providecommand{\BIBentryALTinterwordstretchfactor}{4}
\providecommand{\BIBentryALTinterwordspacing}{\spaceskip=\fontdimen2\font plus
\BIBentryALTinterwordstretchfactor\fontdimen3\font minus
  \fontdimen4\font\relax}
\providecommand{\BIBforeignlanguage}[2]{{%
\expandafter\ifx\csname l@#1\endcsname\relax
\typeout{** WARNING: IEEEtran.bst: No hyphenation pattern has been}%
\typeout{** loaded for the language `#1'. Using the pattern for}%
\typeout{** the default language instead.}%
\else
\language=\csname l@#1\endcsname
\fi
#2}}
\providecommand{\BIBdecl}{\relax}
\BIBdecl

\bibitem{ji_lockdown_2020}
\BIBentryALTinterwordspacing
T.~Ji, H.-L. Chen, J.~Xu, L.-N. Wu, J.-J. Li, K.~Chen, and G.~Qin, ``Lockdown
  {Contained} the {Spread} of 2019 {Novel} {Coronavirus} {Disease} in
  {Huangshi} {City}, {China}: {Early} {Epidemiological} {Findings},''
  \emph{Clinical Infectious Diseases}, vol.~71, no.~6, pp. 1454--1460, Sep.
  2020.  \url{https://doi.org/10.1093/cid/ciaa390}
\BIBentrySTDinterwordspacing

\bibitem{nowzari_analysis_2016}
C.~Nowzari, V.~M. Preciado, and G.~J. Pappas, ``Analysis and {Control} of
  {Epidemics}: {A} {Survey} of {Spreading} {Processes} on {Complex}
  {Networks},'' \emph{IEEE Control Systems Magazine}, vol.~36, no.~1, pp.
  26--46, Feb. 2016.

\bibitem{drakopoulos_efficient_2014}
K.~Drakopoulos, A.~Ozdaglar, and J.~N. Tsitsiklis, ``An {Efficient} {Curing}
  {Policy} for {Epidemics} on {Graphs},'' \emph{IEEE Transactions on Network
  Science and Engineering}, vol.~1, no.~2, pp. 67--75, Jul. 2014.

\bibitem{chernozhukov_causal_2021}
\BIBentryALTinterwordspacing
V.~Chernozhukov, H.~Kasaha, and P.~Schrimpf, ``Causal {Impact} of {Masks},
  {Policies}, {Behavior} on {Early} {Covid}-19 {Pandemic} in the {U}.{S},''
  \emph{Journal of Econometrics}, vol. 220, no.~1, pp. 23--62, Jan. 2021.
  \url{http://arxiv.org/abs/2005.14168}
\BIBentrySTDinterwordspacing

\bibitem{acemoglu_optimal_2020}
\BIBentryALTinterwordspacing
D.~Acemoglu, V.~Chernozhukov, I.~Werning, and M.~D. Whinston, ``Optimal
  {Targeted} {Lockdowns} in a {Multi}-{Group} {SIR} {Model},'' National Bureau
  of Economic Research, Working {Paper} 27102, May 2020.
  \url{https://www.nber.org/papers/w27102}
\BIBentrySTDinterwordspacing

\bibitem{acemoglu_optimal_2021}
\BIBentryALTinterwordspacing
D.~Acemoglu, A.~Fallah, A.~Giometto, D.~Huttenlocher, A.~Ozdaglar, F.~Parise,
  and S.~Pattathil, ``Optimal adaptive testing for epidemic control: combining
  molecular and serology tests,'' \emph{arXiv:2101.00773}, Jan. 2021.
  \url{http://arxiv.org/abs/2101.00773}
\BIBentrySTDinterwordspacing

\bibitem{birge_controlling_2020}
\BIBentryALTinterwordspacing
J.~R. Birge, O.~Candogan, and Y.~Feng, ``\BIBforeignlanguage{en}{Controlling
  {Epidemic} {Spread}: {Reducing} {Economic} {Losses} with {Targeted}
  {Closures}},'' Social Science Research Network, Rochester, NY, {SSRN}
  {Scholarly} {Paper} ID 3590621, May 2020.
  \url{https://papers.ssrn.com/abstract=3590621}
\BIBentrySTDinterwordspacing

\bibitem{anderson_how_2020}
\BIBentryALTinterwordspacing
R.~M. Anderson, H.~Heesterbeek, D.~Klinkenberg, and T.~D. Hollingsworth,
  ``\BIBforeignlanguage{English}{How will country-based mitigation measures
  influence the course of the {COVID}-19 epidemic?}''
  \emph{\BIBforeignlanguage{English}{The Lancet}}, vol. 395, no. 10228, pp.
  931--934, Mar. 2020.
  \url{https://www.thelancet.com/journals/lancet/article/PIIS0140-6736(20)30567-5/fulltext}
\BIBentrySTDinterwordspacing

\bibitem{de_vericourt_informing_2021}
\BIBentryALTinterwordspacing
F.~de~Véricourt, H.~Gurkan, and S.~Wang, ``Informing the {Public} {About} a
  {Pandemic},'' \emph{Management Science}, vol.~67, no.~10, pp. 6350--6357,
  Oct. 2021.
  \url{https://pubsonline.informs.org/doi/abs/10.1287/mnsc.2021.4016}
\BIBentrySTDinterwordspacing

\bibitem{ely_rotation_2021}
\BIBentryALTinterwordspacing
J.~Ely, A.~Galeotti, and J.~Steiner, ``Rotation as {Contagion} {Mitigation},''
  \emph{Management Science}, vol.~67, no.~5, pp. 3117--3126, May 2021.
  \url{https://pubsonline.informs.org/doi/abs/10.1287/mnsc.2020.3910}
\BIBentrySTDinterwordspacing

\bibitem{hernandez-chanto_contagion_2021}
\BIBentryALTinterwordspacing
A.~Hernandez-Chanto, C.~Oyarzun, and J.~Hedlund,
  ``\BIBforeignlanguage{en}{Contagion {Management} through {Information}
  {Disclosure}},'' Social Science Research Network, Rochester, NY, {SSRN}
  {Scholarly} {Paper} ID 3988157, Dec. 2021.
  \url{https://papers.ssrn.com/abstract=3988157}
\BIBentrySTDinterwordspacing

\bibitem{dworczak_simple_2019}
P.~Dworczak and G.~Martini, ``The simple economics of optimal persuasion,''
  \emph{Journal of Political Economy}, vol. 127, no.~5, pp. 1993--2048, 2019.

\bibitem{kolotilin_optimal_2018}
A.~Kolotilin, ``\BIBforeignlanguage{en}{Optimal information disclosure: {A}
  linear programming approach},'' \emph{\BIBforeignlanguage{en}{Theoretical
  Economics}}, vol.~13, no.~2, pp. 607--635, 2018.

\bibitem{guo_interval_2019}
\BIBentryALTinterwordspacing
Y.~Guo and E.~Shmaya, ``\BIBforeignlanguage{en}{The {Interval} {Structure} of
  {Optimal} {Disclosure}},'' \emph{\BIBforeignlanguage{en}{Econometrica}},
  vol.~87, no.~2, pp. 653--675, 2019.
  \url{https://onlinelibrary.wiley.com/doi/abs/10.3982/ECTA15668}
\BIBentrySTDinterwordspacing

\bibitem{kamenica_bayesian_2011}
\BIBentryALTinterwordspacing
E.~Kamenica and M.~Gentzkow, ``\BIBforeignlanguage{en}{Bayesian
  {Persuasion}},'' \emph{\BIBforeignlanguage{en}{American Economic Review}},
  vol. 101, no.~6, pp. 2590--2615, Oct. 2011.
  \url{https://www.aeaweb.org/articles?id=10.1257/aer.101.6.2590}
\BIBentrySTDinterwordspacing

\bibitem{candogan_information_2020}
\BIBentryALTinterwordspacing
O.~Candogan, ``Information {Design} in {Operations}.''\hskip 1em plus 0.5em
  minus 0.4em\relax INFORMS, Nov. 2020, pp. 176--201.
  \url{https://pubsonline.informs.org/doi/abs/10.1287/educ.2020.0217}
\BIBentrySTDinterwordspacing

\bibitem{bergemann_information_2019}
\BIBentryALTinterwordspacing
D.~Bergemann and S.~Morris, ``\BIBforeignlanguage{en}{Information {Design}: {A}
  {Unified} {Perspective}},'' \emph{\BIBforeignlanguage{en}{Journal of Economic
  Literature}}, vol.~57, no.~1, pp. 44--95, Mar. 2019.
  \url{https://www.aeaweb.org/articles?id=10.1257/jel.20181489}
\BIBentrySTDinterwordspacing

\bibitem{kamenica_bayesian_2019}
\BIBentryALTinterwordspacing
E.~Kamenica, ``Bayesian {Persuasion} and {Information} {Design},'' \emph{Annual
  Review of Economics}, vol.~11, no.~1, pp. 249--272, 2019.
  \url{https://doi.org/10.1146/annurev-economics-080218-025739}
\BIBentrySTDinterwordspacing

\bibitem{gentzkow_rothschild-stiglitz_2016}
\BIBentryALTinterwordspacing
M.~Gentzkow and E.~Kamenica, ``\BIBforeignlanguage{en}{A
  {Rothschild}-{Stiglitz} {Approach} to {Bayesian} {Persuasion}},''
  \emph{\BIBforeignlanguage{en}{American Economic Review}}, vol. 106, no.~5,
  pp. 597--601, May 2016.
  \url{https://www.aeaweb.org/articles?id=10.1257/aer.p20161049}
\BIBentrySTDinterwordspacing

\bibitem{candogan_optimal_2021}
\BIBentryALTinterwordspacing
O.~Candogan and P.~Strack, ``\BIBforeignlanguage{en}{Optimal {Disclosure} of
  {Information} to {Privately} {Informed} {Agents}},'' Social Science Research
  Network, Rochester, NY, {SSRN} {Scholarly} {Paper} ID 3773326, Jan. 2021.
  \url{https://papers.ssrn.com/abstract=3773326}
\BIBentrySTDinterwordspacing

\bibitem{olmez_modeling_2021}
\BIBentryALTinterwordspacing
S.~Y. Olmez, S.~Aggarwal, J.~W. Kim, E.~Miehling, T.~Başar, M.~West, and P.~G.
  Mehta, ``\BIBforeignlanguage{en}{Modeling {Presymptomatic} {Spread} in
  {Epidemics} via {Mean}-{Field} {Games}},'' Nov. 2021.
  \url{https://arxiv.org/abs/2111.10422v1}
\BIBentrySTDinterwordspacing

\bibitem{la_torre_mobility_2022}
\BIBentryALTinterwordspacing
D.~La~Torre, D.~Liuzzi, R.~Maggistro, and S.~Marsiglio,
  ``\BIBforeignlanguage{en}{Mobility {Choices} and {Strategic} {Interactions}
  in a {Two}-{Group} {Macroeconomic}–{Epidemiological} {Model}},''
  \emph{\BIBforeignlanguage{en}{Dynamic Games and Applications}}, vol.~12,
  no.~1, pp. 110--132, Mar. 2022.
  \url{https://doi.org/10.1007/s13235-021-00413-z}
\BIBentrySTDinterwordspacing

\bibitem{cianfanelli_lockdown_2021}
\BIBentryALTinterwordspacing
L.~Cianfanelli, F.~Parise, D.~Acemoglu, G.~Como, and A.~Ozdaglar, ``Lockdown
  interventions in {SIR} model: {Is} the reproduction number the right control
  variable?'' \emph{arXiv:2112.06546 [physics]}, Dec. 2021.
  \url{http://arxiv.org/abs/2112.06546}
\BIBentrySTDinterwordspacing

\bibitem{wu_optimal_2020}
\BIBentryALTinterwordspacing
M.~Wu, D.~Shelar, R.~Gopalakrishnan, and S.~Amin,
  ``\BIBforeignlanguage{en}{Optimal {Testing} {Strategy} for {Containing}
  {COVID}-19: {A} {Case}-{Study} on {Indian} {Migrant} {Worker}
  {Population}},'' Rochester, NY, {SSRN} {Scholarly} {Paper}, Oct. 2020.
  \url{https://papers.ssrn.com/abstract=3703429}
\BIBentrySTDinterwordspacing

\bibitem{ivanov_optimal_2015}
\BIBentryALTinterwordspacing
M.~Ivanov, ``\BIBforeignlanguage{en}{Optimal {Signals} in {Bayesian}
  {Persuasion} {Mechanisms}},'' 2015.
  \url{https://www.semanticscholar.org/paper/Optimal-Signals-in-Bayesian-Persuasion-Mechanisms-Ivanov/837a24aeef9c0d5ce94d4b6cbd4eb201a6dc6c2f}
\BIBentrySTDinterwordspacing

\bibitem{blackwell_theory_1954}
D.~Blackwell and M.~Girshick, ``Theory of {Games} and {Statistical}
  {Decisions}.''\hskip 1em plus 0.5em minus 0.4em\relax New York: John Willey
  and Sons, 1954.

\bibitem{candogan_optimality_2019}
\BIBentryALTinterwordspacing
O.~Candogan, ``\BIBforeignlanguage{en}{Optimality of {Double} {Intervals} in
  {Persuasion}: {A} {Convex} {Programming} {Framework}},'' Social Science
  Research Network, Rochester, NY, {SSRN} {Scholarly} {Paper} ID 3452145, Sep.
  2019.  \url{https://papers.ssrn.com/abstract=3452145}
\BIBentrySTDinterwordspacing

\bibitem{shaked_univariate_2007}
\BIBentryALTinterwordspacing
``\BIBforeignlanguage{en}{Univariate {Stochastic} {Orders}},'' in
  \emph{\BIBforeignlanguage{en}{Stochastic {Orders}}}, ser. Springer {Series}
  in {Statistics}, M.~Shaked and J.~G. Shanthikumar, Eds.\hskip 1em plus 0.5em
  minus 0.4em\relax New York, NY: Springer, 2007, pp. 3--79.
  \url{https://doi.org/10.1007/978-0-387-34675-5_1}
\BIBentrySTDinterwordspacing

\bibitem{ivanov_optimal_2021}
\BIBentryALTinterwordspacing
M.~Ivanov, ``\BIBforeignlanguage{en}{Optimal monotone signals in {Bayesian}
  persuasion mechanisms},'' \emph{\BIBforeignlanguage{en}{Economic Theory}},
  vol.~72, no.~3, pp. 955--1000, Oct. 2021.
  \url{https://doi.org/10.1007/s00199-020-01277-x}
\BIBentrySTDinterwordspacing

\end{thebibliography}

\section{Appendix}

\subsection{Proof of Proposition \ref{prop:eqbm_formula_prop_1}}
\noindent By construction, the groups of agents are implicitly ordered by their preference for working in-person. The following lemma
shows that \textit{under any equilibrium} $\final(i)$ as defined in Definition $\ref{defn:pse}$, there is a critical group $k^*(i) \in [\K]$ such that all agents of groups $l$ with smaller benefits $v_l < v_{k^*(i)}$ work remotely, and all agents of groups $l$ with larger benefits $v_l > v_{k^*(i)}$ work in-person.

\begin{lemma}
\label{lemma:rectangular_curve}
For any equilibrium $\final(i)$, there exists a critical group $k^*(i) \in [\K]$ such that for all $l < k^*(i)$, $y_l^*(i) = 0$ and all $l > k^*(i)$ $y_l^*(i) = x_l$.
\end{lemma}
\begin{proof}
Our proof proceeds by construction. Suppose that $\final(i)$ is an equilibrium and let $k^*(i) = \min\{i: y_i^*(i) > 0\}$. Suppose, by contradiction, that there exists $l > k^*(i)$ such that $y_l^*(i) < x_l$. Then, $v_l \geq c_1(\|\fin\|)\theta_i+c_2(\|\fin\|)$. However, since $v_{k^*(i)} > v_l$, this implies that $v_{k^*(i)} > c_1(\|\fin\|)\theta_i+c_2(\|\fin\|)$ and by Definition $\ref{defn:pse}$, $y_{k^*(i)}(i) = 0$. This is a contradiction, so it is proven that for all $l > k^*(i)$, $y_l^*(i) = x_l$ and hence $k^*(i)$ satisfies the conditions of the critical group.
\end{proof}

\noindent The next lemma precisely computes the in-person equilibrium mass in response to signal $i$. Given the restriction on $\final(i)$ established in the previous lemma, this precisely identifies the equilibrium mass vector $\final(i)$ that the agents will achieve as in the proposition statement and proves the result. 
\begin{lemma}
\label{lemma:compute_inperson_mass}
In any equilibrium $\final(i)$, the equilibrium in-person mass $ m(\theta_i) \coloneqq 1-\norm{\final(i)} = \sup\{u: v(u) \geq c_1(1-u)\theta_i+c_2(1-u)\}$.
\end{lemma}
\begin{proof}
Suppose $m(\theta_i)  = z$ with $s_{j-1} < z < s_j$ for some $j$, then by Definition \ref{defn:pse},  $v_j = c_1(1-z)\theta_i+c_2(1-z)$. Notice, $v_j = v(z) = c_1(1-z)\theta_i+c_2(1-z)$ and that $v(u)$ is non-increasing and $c_1(1-u)\theta_i+c_2(1-u)$ is strictly increasing in $u$, so $z = \sup\{u: v(u) \geq c_1(1-u)\theta_i+c_2(1-u)\}$.\\
Suppose $m(\theta_i)  = z$ with $z = s_j$ for some $j$. Then $v(z) = v_{j+1}$ and for all $\epsilon < x_j$, $v(z-\epsilon) = v_{j}$. But by Definition \ref{defn:pse}, $v(z) = v_{j+1} \leq c_1(1-z)\theta_i+c_2(1-z)$ and for all $\epsilon < x_j$, $v(z-\epsilon) = v_{j} \geq c_1(1-z+\epsilon)\theta_i+c_2(1-z+\epsilon)$. But again since $c_1(1-u)\theta_i+c_2(1-u)$ is strictly increasing in $u$, this implies that $z = \sup\{u: v(u) \geq c_1(1-u)\theta_i+c_2(1-u)\}$.
\end{proof}

\subsection{Proof of Lemma
\ref{lemma:m_smooth}}
\noindent Observe that $m(\theta) \leq 1$ since for all $u > 1$, $v(u) = 0$ and $c_1(1-u) \theta + c_2(1-u) > 0$. Similarly, since $v(0) \geq 0 = c_1(1) \theta + c_2(1)$, $m(\theta) \geq 0$ so clearly $m(\theta)$ is bounded.

\noindent Similarly, letting $f(u) = \frac{v(u)-c_2(1-u)}{c_1(1-u)}$, observe $m(\theta) = \sup\{u: f(u) \geq \theta \}$. Since $c_1(1-u)$ is a strictly increasing function in $u$ and $v(u)-c_2(1-u)$ is a non-increasing function, $f(\cdot)$ is monotonically decreasing. For any $\theta' \leq \theta''$ notice $\{u: f(u) \geq \theta'' \} \subseteq \{u: f(u) \geq \theta' \}$, so $m(\theta'') \leq m(\theta')$ and hence $m$ is non-increasing. Moreover, applying Berge's Maximum Principle, we can determine that $m(\theta)$ is continuous.$\blacksquare$

\subsection{Proof of Lemma \ref{lemma:intervals}}
The proof is quite immediate since notice that $\curve = \cup_{i=1}^{\K}  \{z(\init,u): s_{k-1} \leq u \leq s_k\}$. Consequently, $\goal \cap \curve = \cup_{i=1}^{\K}  \big(\goal \cap \{z(\init,u): s_{k-1} \leq u \leq s_k\}\big)$.
However, notice $\{z(\init,u): s_{k-1} \leq u \leq s_k\}$ is a line segment in $\R^\K$ so its intersection with a convex, compact set $\goal \subset \R^\K$ is just another line segment $\{z(\init,u): \tilde{\omega}_k^1  \leq u \leq \tilde{\omega}_k^2\}$ for some constants $\tilde{\omega}_k^1, \tilde{\omega}_k^2$. Eliminating empty intervals and condensing the intervals if $\tilde{\omega}_{j}^2 = \tilde{\omega}_{j+1}^1$ for any $j$ yields the result.

\subsection{Proof of Lemma \ref{lemma:num_signals_bdd}}
The proof follows directly from Lemma 1 in \cite{candogan_optimality_2019}.
\subsection{Proof of Theorem \ref{thm:stateless}}

(R1) The proof for distributions $F$ in R1 follows from construction as $\pi$ induces $\mathcal{T}_\pi = \{(1,\mu)\}$ and from Equation $\eqref{eqn:red_1_form}$ this implies that $V^* = 1$ which is maximal. \\

(R2a) The proof for distributions $F$ in R2a follows from construction as $\pi$ induces $\mathcal{T}_\pi = \{(\lambda F(t) +  \alpha (1-F(t)),\delta(\alpha,\lambda,t)),((1-\lambda) F(t) +  (1-\alpha) (1-F(t)), \delta(1-\alpha,1-\lambda,t) )\}$ and from Equation $\eqref{eqn:red_1_form}$ this implies that $V^* = 1$ which is maximal. \\

(R3-4) \\
\noindent The proofs for R3 and R4 are symmetric (can modify all objects in $\Theta$ by going from $\theta$ to $M-\theta$ and the proof is analogous) so we only prove for R3. Moreover, without loss of generality, assume $M=1$ as an optimal signalling mechanism is invariant to linear scaling.\\
Observe that we aim to solve:
\begin{align*}
V^* = \max_{\pi: \langle \{g_\theta\}_{\theta \in \Theta},\I\rangle} \sum_{k=1}^{\KK} q_i \mathbb{I}\{\lol{k} \leq \theta_i \leq \hil{k}\}
\end{align*}
However, using the established properties of signalling mechanisms, we can search over the distributions of the posterior means $G$ by searching over the space of all mean-preserving contractions of $F$, $G \succcurlyeq F$. Recall that $G \succcurlyeq F$ if and only if $\int_{0}^{x} (1-G(t))dt \geq \int_{0}^{x} (1-F(t))dt$ for all $x \in [0,1]$ with equality at $x=1$. In the quantile space, this is equivalent to $\int_{0}^{x} G^{-1}(t)dt \geq \int_{0}^{x} F^{-1}(t)dt$ for all $x \in [0,1]$ with equality at $x=1$. We will refer to this constraint as the MPC constraint.\\
Moreover, using Lemma \ref{lemma:num_signals_bdd} that we can further restrict $G$ to be a discrete distribution using $\KK+1$ mass points (each corresponding to a posterior mean) with $\theta_i \in [\lol{i},\hil{i}]$ with probability $q_i$ and $\theta_{\KK+1}$ with probability $q_{\KK+1}$. Since the posterior mean $\theta_i$ is induced with probability $q_i$ $\mu > \hil{\KK}$ and $\sum_{i=1}^{\KK+1} q_i \theta_i = \mu$. However, since $\theta_i <  \hil{\KK} < \mu$ for all $i \in [\KK]$, this implies that $\theta_{\KK+1} > \mu$ and $q_{\KK+1} > 0$.  Hence, we can re-express the objective as follows:
\begin{align*}
V^* &= \max_{G \succcurlyeq F;\forall i \in [\KK] \mathbb{P}_{\theta \sim G}[\theta = \theta_i] = q_i>0\text{ for unique }\theta_i \in [\lol{i},\hil{i}]} \sum_{k=1}^{\KK} q_i \\
&= \min_{G \succcurlyeq F;\forall i \in [\KK] \mathbb{P}_{\theta \sim G}[\theta = \theta_i] = q_i>0\text{ for unique }\theta_i \in [\lol{i},\hil{i}]} q_{\KK+1}
\end{align*}
Observe that $G$ is parameterized by $(q_i,\theta_i)_{i \in [\KK+1]}$, so we can equivalently express the condition that $G \succcurlyeq \dist$ as in Lemma 4 of \cite{candogan_optimal_2021} by adding constraints that enforce the mean-preserving contraction results: 
 \begin{align}
 &\min q_{\KK+1}               \\
 \text{s.t. }&\sum_{i=1}^{\KK+1} q_i = 1\\
 &\sum_{i=1}^{\KK+1} q_i \theta_i = \mu \label{eqn:mean_mean}\\
 & \lol{l} \leq \theta_l \leq \hil{l} \hspace{1cm} \forall l \in [\KK]\\
 & \mu \leq \theta_{\KK+1} \leq 1\\
 & \sum_{j=1}^{m} q_j\theta_j \geq \int_0^{\sum_{j=1}^{m} q_j} \dist^{-1}(t)dt\hspace{0.25cm} \forall m \in [\KK] \label{eqn:MPC}
\end{align}
 We can reduce the enforcement of the MPC constraints in Equation \eqref{eqn:MPC} from $\int_0^s \dist^{-1}(t)dt \leq \int_0^x G^{-1}(t)dt$ for all $x \in [0,1]$ to merely checking for $x \in \mathcal{S} \coloneqq \{\sum_{j=1}^m q_j: m \in [\KK+1]\}$. This is done without loss since $\int_0^s \dist^{-1}(t)dt$ is convex and $\int_0^s G^{-1}(t)dt$ is piecewise-linear with breakpoints at these $x \in \mathcal{S}$ so the remaining constraints are implied.\\
At this point, we make a brief comment on the optimal signalling mechanism for the general regime of (R2). Observe that for (R3) we can sum in order of index in Equation \eqref{eqn:MPC} since $\theta_i$ is increasing for all $i \in [\KK+1]$. This need not necessarily be true in (R2) as the posterior mean that lies outside these intervals could be positioned anywhere relative to the other posterior means depending on the characteristics of the distribution. However, by iterating over all $\KK+1$ possible placements of $\theta_{\KK+1}$ relative to $\{\theta_i\}_{i\in[\KK]}$, this program will solve for the optimal direct signalling mechanism for (R2) by taking the resultant $\{(q_i^*,\theta_i^*)\}_{i \in [\KK+1]}$. While the formulation is nonconvex as written, introducing variables $z_i \coloneqq q_i\theta_i$ as in \cite{candogan_optimal_2021} can convexify the form and make this program efficiently solvable.\\
For (R3), however, we can continue to simplify. Recall that $q_{\KK+1}^* > 0$ and next suppose that $q_j^* > 0$ for some $j < \KK$. This is a contradiction; consider decreasing $q_j^*$ by $(\frac{\theta_{\KK+1}-\theta_{\KK}}{\theta_{\KK+1}-\theta_{j}})\epsilon$, increasing $q_{\KK}^*$ by $\epsilon$, decreasing $q_{\KK+1}^*$ by $(\frac{\theta_{\KK}-\theta_{j}}{\theta_{\KK+1}-\theta_{j}})\epsilon$. Then, observe that the objective is strictly improved and all the non-MPC constraints are still obeyed by design. Moreover, notice using the alternative definition of mean-preserving contraction in Theorem 3.A.7. of \cite{shaked_univariate_2007}, this convex stochastic modification creates a mean-preserving contraction of the $G^*$ that implements $\{(q_i^*,\theta_i^*)\}$ so by transitivity it must be a mean-preserving contraction of $F$ and satisfy the MPC constraints in Equation \eqref{eqn:MPC}. \\
Renaming $q_{\KK}$ to $q_1$ and $q_{\KK+1}$ to $q_2$, we can reduce the optimization to:
 \begin{align}
 &\max q_1             \\
 \text{s.t. }& q_1 \theta_{1} + q_2 \theta_2 = \mu \\
 & q_1 + q_2 = 1 \\
 & q_1\theta_{1} \geq \int_0^{q_1} \dist^{-1}(t)dt \label{eqn:convv} \\
  & \lol{\KK} \leq  \theta_1 \leq \hil{\KK} \\
 & \mu \leq \theta_{2} \leq 1
\end{align}
 Suppose we find a solution $q_1^*, q_2^*, \theta_1^*, \theta_2^*$ where $\theta_{1}^* < \hil{\KK}$. Observe that we can find another solution with the same objective by choosing $q_1' = q_1^*, q_2' = q_2^*$, $\theta_{1}^{'} = \hil{\KK}$ and $\theta_{2}^{'} = \mu+\frac{(\mu-\hil{\KK})(\theta_2^* - \mu)}{(\mu-\theta_1^*)}$. Observe that this solution obeys all the previous constraints. Consequently, we can restrict $\theta_1 = \hil{\KK}$ without loss.\\
Observing that $f(x) = \int_0^{x} \dist^{-1}(t)dt$ is convex with $f(0) = 0$ implies that $\theta x \geq f(x)$ for all $x \in [0,h(\theta)]$. Replacing $\theta_{\KK}$ with $\hil{\KK}$, $q_2$ with $1-q_1$ and then the constraint in \eqref{eqn:convv} with $0 \leq q_1 \leq h(\hil{\KK})$ results in the following optimization:
\begin{align*}
    &\max_{q_1,\theta_2} q_1 \\
    \text{ s.t.   } &q_1 \hil{\KK} + (1-q_1) \theta_2 = \mean\\
    & 0 \leq q_1 \leq h(\hil{\KK})\\
  & \mu \leq \theta_2 \leq M
\end{align*}
We then eliminate $\theta_2$ by using the first constraint as $\theta_2 = \frac{\mu-q_1 \hil{\KK}}{1-q_1}$. Since $\mu \geq \hil{\KK}$, $\mu \leq \theta_2 \leq M$ simplifies to $q_1 \leq \frac{M-\mu}{M-\hil{\KK}}$ and we admit the solution that $q_1^* = \min\{h(\hil{\KK}),\frac{M-\mu}{M-\hil{\KK}}\}$.

\noindent Finally, given the objective value $V^*$ returned, we solve the dual problem of finding the mechanism $\langle \{g_\theta\}_{\theta \in \Theta},\I \rangle$ that results in $V^*$. To solve this, we appeal to the results in Remark 1 and Proposition 1 in \cite{gentzkow_rothschild-stiglitz_2016} as well as Theorem 3.A.4 in \cite{shaked_univariate_2007}. Consider the discrete distribution $G$ that places probability $q_1^*$ on $\hil{\KK}$ and $1-q_1^*$ on $\theta_2^*$. Then, the corresponding function $g(x) = \int_{0}^x G(t) dt$ defined by:
\begin{align}
    g(x) &= \begin{cases} 0 & x \leq \hil{\KK} \\ q_1^* (x-\hil{\KK}) & \hil{\KK} < x \leq \theta_2^* \\ q_1^*(\theta_2^*-\hil{\KK})+(x-\theta_2^*) & x > \theta_2^* \end{cases}
\end{align}
It is known that $g$ is convex and that, for $f(x) = \int_{0}^x F(t) dt$, $g(x) \leq f(x)$ for all $x \in [0,1]$. Observe that $g'(0) = f'(0)$ and $g'(1) = f'(1)$. From \cite{shaked_univariate_2007} and \cite{ivanov_optimal_2021}, it is known that if there exists $s \in [\hil{\KK},\theta_2^*]$ such that $f(s) = g(s)$, then $g$ is tangent to $f$ at $s$ and therefore is tangent on each linear segment of $g$. It is then a result that the mechanism can then be implemented using a monotone partitional signalling mechanism with $t_0 = 0, t_1 = s, t_2 = 1$ with $s = F^{-1}(q_1^*)$. This is consistent since the probability signal $1$ is generated is represented by $q_1^*$ and $\mathbb{P}[\theta \in [0,s]] = F(s) = F(F^{-1}(q_1^*)) = q_1^*$.

\noindent To show this holds in this setting, suppose by contradiction that for all $s \in [\hil{\KK},\theta_2^*]$, $f(s) > g(s)$. Then since $f-g$ is convex over $[\hil{\KK},\theta_2^*]$, let $\inf_{t \in [\hil{\KK},\theta_2^*]} f(t)-g(t) = \epsilon > 0$ with some minimizer $t^* \in [\hil{\KK},\theta_2^*]$ such that $f(t^*)-g(t^*) = \epsilon$. Moreover, let $\tilde{\theta_2^*}$ solve $(q_1^*+\epsilon)(x-\hil{\KK}) = q_1^*(\theta_2^*-\hil{\KK})+(x-\theta_2^*)$. Then consider the function $\tilde{g}$ as follows:
\begin{align*}
    \tilde{g}(x) &= \begin{cases} 0 & x \leq \hil{\KK} \\ (q_1^*+\epsilon) (x-\hil{\KK}) & \hil{\KK} < x \leq \tilde{\theta_2^*} \\ (q_1^*+\epsilon)(\tilde{\theta_2^*}-\hil{\KK})+(x-\tilde{\theta_2^*}) & x > \tilde{\theta_2^*} \end{cases}
\end{align*}
Observe that $\tilde{g}$ is still convex and that $g\leq\tilde{g}\leq f$, so by \cite{gentzkow_rothschild-stiglitz_2016} the distribution over posterior means with $q_1^*+\epsilon$ on $\hil{\KK}$ and $1-q_1^*-\epsilon$ on $\tilde{\theta_2^*}$ is implementable through a signalling mechanism. However, this violates the optimality of $q_1^*$ and is a contradiction. Hence, there must exist an  $s \in [\hil{\KK},\theta_2^*]$ such that $f(s) = g(s)$.
\subsection{Proof of Proposition \ref{prop:stateful}}
\noindent Observe that we can expand the optimization in \eqref{eqn:stateful_b} to include the constraints as follows:
\begin{align*}
    V^* &= \max_{\langle \{g_\nu\}_{\nu \in \Theta},\I\rangle} \sum_{j=1}^N \sum_{i=1}^{|\I|} p_j g_{\nu_j}(i) \mathbb{I}\{c_j \leq \theta_i \leq \infty\} \\
    &\text{s.t. }\sum_{i=1}^{|\I|} g_{\nu_j}(i) = 1, \forall j\in[N] \\
    &\hspace{1cm} 0 \leq g_{\nu_j}(i) \leq 1, \forall i\in |\I|,j\in[N] \\
    &\hspace{1cm}\theta_i = \frac{\sum_{j=1}^N \nu_j p_j g_{\nu_j}(i)}{\sum_{j=1}^N p_j g_{\nu_j}(i)}, \hspace{0.5cm} \forall i \in |\I|
\end{align*}
As mentioned, analogous to Lemma \ref{lemma:num_signals_bdd}, if two signals $i_1, i_2$ are such that there exists $\gamma_{j} \leq \theta_{i_1}, \theta_{i_2} < \gamma_{j+1}$, we can merge the signals without loss. Therefore, for each interval $[\gamma_{i-1},\gamma_{i})$ with $i\in [N+1]$, at most one posterior mean need be present so $\I = [N+1]$ with $\theta_i \in [\gamma_{i-1},\gamma_{i})$ for all $i$. 
Reducing by observing that $\theta_i \geq \gamma_j$ for all $i \geq j+1$:
\begin{align}
    V^* &= \max_{\{ g_{\nu_j}(i)\}} \sum_{j=1}^N \sum_{i=j+1}^{N+1} p_j g_{\nu_j}(i) \\
    &\text{s.t. }\sum_{i=1}^{N+1}  g_{\nu_j}(i) = 1, \hspace{0.5cm}\forall j\\
    &\hspace{1cm} 0 \leq g_{\nu_j}(i) \leq 1,\hspace{0.5cm} \forall i,j \\
    &\hspace{1cm} \gamma_{i-1} \leq  \frac{\sum_{j=1}^N \nu_j p_j g_{\nu_j}(i)}{\sum_{j=1}^N p_j g_{\nu_j}(i)} \leq \gamma_{i}, \hspace{0.5cm} \forall i\in[N+1]
\end{align}
By renaming $z_{ji} = p_j g_{\nu_j}(i)$ and expanding the final constraint, we arrive at the linear program.
\end{document}